\newif\ifllncs
\llncsfalse

\ifllncs
\documentclass{llncs}
\else
\documentclass{article}
\fi

\ifllncs
\usepackage{amsmath}
\else
\usepackage{amsthm, amsmath}
\usepackage[margin=1in]{geometry}
\fi
\usepackage{amssymb}

\usepackage{thmtools,thm-restate}
\ifllncs
\else
\declaretheorem{theorem}
\fi

\usepackage{hyperref}
\usepackage{enumitem}

\ifllncs
\spnewtheorem{fact}{Fact}{\bfseries}{\rmfamily}
\spnewtheorem{hypothesis}{Hypothesis}{\bfseries}{\rmfamily}
\else
\newcommand{\email}[1]{\href{mailto:#1}{\nolinkurl{#1}}}

\theoremstyle{plain}

\newtheorem{lemma}[theorem]{Lemma}

\theoremstyle{definition}
\newtheorem{definition}[theorem]{Definition}
\newtheorem{example}[theorem]{Example}

\newtheorem{hypothesis}[theorem]{Hypothesis}
\newtheorem{fact}[theorem]{Fact}

\newtheorem{corollary}[subsection]{Corollary}

\theoremstyle{remark}

\fi

\newcommand{\Z}{\ensuremath{\mathbb{Z}}}

\newcommand{\R}{\ensuremath{\mathbb{R}}}
\newcommand{\Q}{\ensuremath{\mathbb{Q}}}
\newcommand{\C}{\ensuremath{\mathbb{C}}}
\newcommand{\primes}{\ensuremath{\text{primes}}}

\newcommand{\defined}{\ensuremath{\stackrel{\text{def}}{=}}}

\newcommand{\degree}{\ensuremath{\operatorname{deg}}}

\newcommand{\ie}{\emph{i.e.,}}

\newcommand{\ZZ}[1]{\Z/#1\Z}
\newcommand{\ZZstar}[1]{(\Z/#1\Z)^*}

\newcommand{\inbraces}[1]{\ensuremath{ \left\{ #1 \right\} } }
\newcommand{\inset}[1]{\ensuremath{ \left\{ #1 \right\} } }
\newcommand{\inparen}[1]{\ensuremath{ \left( #1 \right) } }

\newcommand{\Qp}{\ensuremath{\mathbb{Q}_p}}
\newcommand{\Qpbar}{\ensuremath{\overline{\mathbb{Q}}_p}}
\newcommand{\Qbar}{\ensuremath{\overline{\mathbb{Q}}}}
\newcommand{\Zbar}{\ensuremath{\overline{\mathbb{Z}}}}
\newcommand{\suchthat}{\ensuremath{~\middle|~}}
\newcommand{\E}{\ensuremath{\mathbb{E}}}
\newcommand{\pnorm}[1]{\ensuremath{\left| #1 \right|_p}}
\newcommand{\ppnorm}[2]{\ensuremath{\left| #1 \right|_{#2}}}

\title{Cryptographic applications of capacity theory: On the optimality of Coppersmith's method for univariate polynomials}

\author{Ted Chinburg \\ \email{ted@math.upenn.edu}
\and
Brett Hemenway \\ \email{fbrett@cis.upenn.edu}
\and
Nadia Heninger \\ \email{nadiah@cis.upenn.edu}
\and
Zachary Scherr \\ \email{zscherr@math.upenn.edu}}

\date{\today}

\begin{document}

\maketitle

\begin{abstract}
We draw a new connection between Coppersmith's method for finding
small solutions to polynomial congruences modulo integers and the
capacity theory of adelic subsets of algebraic curves.  Coppersmith's
method uses lattice basis reduction to construct an auxiliary
polynomial that vanishes at the desired solutions.  Capacity theory
provides a toolkit for proving when polynomials with certain
boundedness properties do or do not exist.  Using capacity theory, we
prove that Coppersmith's bound for univariate polynomials is optimal
in the sense that there are \emph{no} auxiliary polynomials of the
type he used that would allow finding roots of size $N^{1/d+\epsilon}$
for monic degree-$d$ polynomials modulo $N$.  Our results rule out the
existence of polynomials of any degree and do not rely on lattice
algorithms, thus eliminating the possibility of even
superpolynomial-time improvements to Coppersmith's bound.  We extend
this result to constructions of auxiliary polynomials using binomial
polynomials, and rule out the existence of any auxiliary polynomial of
this form that would find solutions of size $N^{1/d+\epsilon}$ unless
$N$ has a very small prime factor.
\end{abstract}

\ifllncs
\begin{keywords}
Coppersmith's method, lattices, polynomial congruences, capacity theory, RSA
  \end{keywords}
  \fi

\section{Introduction}

Coppersmith's method \cite{C97,CoppersmithFinding} is a celebrated technique in public-key cryptanalysis for finding small roots of polynomial equations modulo integers.  In the simplest case, one is given a degree-$d$ monic polynomial $f(x)$ with integer coefficients, and one wishes to find the integers $r$ modulo a given integer $N$ for which $f(r) \equiv 0 \bmod N$.  When $N$ is prime, this problem can be efficiently solved in polynomial time, but for composite $N$ of unknown factorization, no efficient method is known in general.  
In fact, such an algorithm would immediately break the RSA cryptosystem, by allowing one to decrypt ciphertexts $c$ by finding roots of the polynomial $f(x) = x^e - c \bmod N$.

While it appears intractable to solve this problem in polynomial time, Coppersmith showed that one can efficiently find all \emph{small} integers $r$ such that
$f(r) \equiv 0$ mod $N$.
More precisely, he proved the following result in \cite{C97}:

\begin{theorem}[Coppersmith 1996]
	\label{thm:Cop1}
	Suppose one is given a modulus $N$ and a monic polynomial $f(x) = x^d + f_{d-1} x^{d-1} + \cdots + f_1 x + f_0$ in  $ \Z[x]$. One can find all $r \in \Z$ such that
	\begin{equation}
          \label{eq:Cop1}
        |r| \le N^{1/d}\quad \mathrm{and} \quad  f(r) \equiv 0 \bmod N
        \end{equation}
        in polynomial time in $\mathrm{log}(N) + \sum_i \log |f_i|$.
\end{theorem}

The algorithm he developed to prove this result has applications across public-key cryptography, including cryptanalysis of low public exponent RSA with fixed-pattern or affine padding~\cite{C97}, the security proof of RSA-OAEP~\cite{S01}, and showing that the least significant bits of RSA are hardcore~\cite{SPW06a}.  We discuss these applications in more detail in \S \ref{sec:applications}.
If the exponent $1/d$ in the bound in Equation~\ref{eq:Cop1} could be increased, it would have immediate 
practical impact on the security of a variety of different cryptosystems.

In followup work, \cite[\S 4]{CoppersmithFinding} Coppersmith
speculates about possible improvements of this exponent $1/d$. The
main conclusion of \cite[\S 4]{CoppersmithFinding} is that ``We have
tried to abuse this method to obtain information that should otherwise
be hard to get, and we always fail.''
Later, the hardness of finding roots of $f(x)$ of size $N^{1/d+\epsilon}$ for $\epsilon > 0$ 
was formalized as a concrete cryptographic hardness assumption \cite{SPW06a}.

Coppersmith's proof of Theorem \ref{thm:Cop1} relies on constructing a polynomial 
$h(x)$ such that any small integer $r$ satisfying $f(r) \equiv 0 \bmod N$ is a root of $h(x)$ over the integers.
He finds
such an auxiliary polynomial $h(x)$ by constructing a basis for a lattice of polynomials, and then by using the Lenstra-Lenstra-Lovasz lattice basis reduction algorithm \cite{lll} to find a ``small'' polynomial in this lattice.  The smallness condition ensures that any small integer $r$ satisfying $f(r) \equiv 0$ mod $N$ must be a root of $h(x)$.  
The algorithm then checks which rational roots $r$ of $h(x)$ have the desired properties.

\paragraph{Our Results.}

In this paper, our main result is that one cannot increase the exponent $1/d$ in Coppersmith's theorem by using auxiliary polynomials of the kind he considers.  This eliminates possible improvements to the method using improvements in lattice algorithms or shortest vector bounds.  We obtain our results by drawing a new connection between this family of cryptographic techniques and results from the \emph{capacity theory} of adelic subsets of algebraic curves.  We will use fundamental results of Cantor \cite{Cantor}
and Rumely \cite{Rumely1,Rumely2} about capacity theory to prove several results about such polynomials.  
\medbreak
In particular, we will prove in Theorem~\ref{thm:exist} a stronger form of the following result.  This result shows that there are \emph{no} polynomials of the type used by Coppersmith that could lead to an improvement of the bound in (\ref{eq:Cop1}) from $N^{1/d}$ to $N^{1/d+\epsilon}$ for any $\epsilon > 0$.

\begin{restatable}[Optimality of Coppersmith's Theorem]{theorem}{optimality}
\label{thm:epsrsult}
Suppose $\epsilon > 0$.  There does not exist a non-zero polynomial $h(x) \in \Q[x]$ of the form 
\begin{equation}
\label{eq:form}
h(x) = \sum_{i,j \ge 0} a_{i,j} \ x^i \ (f(x)/N)^j
\end{equation}with $a_{i,j} \in \mathbb{Z}$  such that $|h(z)| < 1$ for all $z$ in the
complex disk $\{z \in \C : |z| \le N^{(1/d) + \epsilon} \}$. Furthermore, if $\epsilon > \ln(2)/\ln(N)$ there is no
such $h(x)$ such that $|h(z)| < 1$ for all $z$ in the real interval $[-N^{1/d + \epsilon}, N^{1/d +\epsilon}]$.
\end{restatable}

Note that in order
for Coppersmith's method to run in polynomial time, $h(x)$ should have degree bounded by a polynomial in $\ln(N)$.  Theorem \ref{thm:epsrsult} says that
when $\epsilon > 0$ there are no polynomials of \emph{any} degree satisfying the stated bounds.  We can thus eliminate the possibility of an improvement to this method with even superpolynomial running time.

In \cite{CoppersmithFinding}, Coppersmith already noted that it did not appear possible
to improve the exponent $1/d$ in his result by searching for roots in the real interval $[-N^{1/d+\epsilon},N^{1/d +\epsilon}]$
instead of in the complex disk of radius $N^{1/d+\epsilon}$.  
The last statement in Theorem \ref{thm:epsrsult} quantifies this observation, since $\ln(2)/\ln(N) \to 0$
as $N \to \infty$.



Coppersmith also notes that since the
binomial polynomials
\[
b_i(x) = x \cdot (x-1) \cdots (x - i + 1)/i!
\]
take integral values on integers, one could replace $x^i$ in (\ref{eq:form}) by $b_i(x)$ and $(f(x)/N)^j$ by $b_j(f(x)/N)$.
 Coppersmith observed (backed up by experiments) that this leads to a small improvement on the size of the root that can be found, and a speedup for practical computations.  The improvement is proportional to the degree of the auxiliary polynomial $h(x)$ that is constructed, and is thus limited for a polynomial-time algorithm.
 
We show that the exponent $1/d$ in Coppersmith's theorem still cannot
be improved using binomial polynomials, but for a different reason.  Our results come in two parts.   First, we show that the exact analogue of Theorem~\ref{thm:epsrsult} is false in the case of integral combinations of binomial polynomials.  In fact, there are such combinations that have all the properties required in the proof of Coppersmith's theorem.  The problem is that these polynomials have very large degree, and in fact, they vanish at \emph{every} small integer, not just the solutions of the congruence.  This is formalized in the following theorem, which is a simplified version of Theorem \ref{thm:epsrsultpos}.

\begin{theorem}[Existence of binomial auxiliary polynomials]
\label{thm:epsrsultposeasy}Suppose $\delta$ is any positive real number.  For all sufficiently large
integers $N$ there is a non-zero polynomial of the form $h(x) = \sum_i a_{i} \ b_i(x) $ 
with $a_i \in \mathbb{Z}$ such that $|h(z)| < 1$ for all 
$z$ in the complex disk $\{z \in \C : |z| \le N^{\delta} \}$. 
\end{theorem}

Second, we show that the existence of these polynomials still does not permit cryptographically useful improvements to Coppersmith's bound beyond $N^{1/d}$.
 This is because if one is able to use binomial polynomials of small degree to obtain such an improvement, then
 the modulus $N$ must have a small prime factor.   In that case, it would have been more efficient to factor $N$ and use the factorization to find the roots.  More precisely, we will show in Theorem \ref{thm:epsresultneg} a stronger form of the following result:

\begin{theorem}[Negative Coppersmith Theorem for binomial polynomials]
\label{thm:epsresultnegeasy}  Suppose $\epsilon > 0$ and that $M $ and $N$ are  integers 
with $1.48774 N^\epsilon \ge M \ge 319$.  If there is 
a non-zero polynomial $h(x)$ of the form 
\begin{equation}
\label{eq:binomialformnegeasy}
h(x) = \sum_{0 \le i,j \le M} a_{i,j} \ b_i(x) \ b_j(f(x)/N)
\end{equation}
with $a_{i,j} \in \mathbb{Z}$ such that $|h(z)| < 1$ for $z$ in the complex
disk  $\{z \in \C : |z| \le N^{1/d+ \epsilon} \}$, then $N$ must have 
a prime factor less than or equal to $M$.  In particular, this will be the case for all large $N$
if we let $M = \ln(N)^c$ for some fixed integer $c > 0$. 
\end{theorem}

Note that the integer $M$ quantifies ``smallness'' in Theorem~\ref{thm:epsresultnegeasy} in two ways.  First, it is a bound on the degree of the binomial polynomials that are allowed to be used to create auxiliary polynomials.  But then if a useful auxiliary polynomial exists, then $N$ must have a factor of size less than or equal to $M$.  As a special case of Theorem~\ref{thm:epsresultnegeasy}, if $N = pq$ is an RSA modulus with two large equal sized prime factors, then any auxiliary polynomial of the form in \eqref{eq:binomialformnegeasy} that can find roots of size $N^{1/d+\epsilon}$ must involve binomial terms with $i$ or $j$ at least $1.48774N^{\epsilon}$.

Note that Coppersmith's theorem in its original form is not sensitive to whether or not $N$ has small prime factors.  Theorem~\ref{thm:epsresultnegeasy} shows that the existence of useful auxiliary polynomials \emph{does} depend on whether $N$ has such small factors.  



The paper is organized in the following way.  In \S \ref{s:copper} we begin by recalling Coppersmith's algorithm for finding small solutions of polynomial congruences.  In \S \ref{sec:applications}
 we recall some mathematical hardness assumptions and we discuss their connection to the security
of various cryptosystems and Coppersmith's algorithm. 
In \S \ref{sec:capacity} we review some basic notions from algebraic number theory, and we recall
some results of Cantor \cite{Cantor} and Rumely \cite{Rumely1,Rumely2} on which our work is based.
At the end of \S \ref{sec:capacity} we prove Theorem \ref{thm:exist}, which implies 
Theorem \ref{thm:epsrsult}.  We state and prove Theorem \ref{thm:epsrsultpos} and
Theorem \ref{thm:epsresultneg} in \S \ref{s:improve};  these imply Theorem \ref{thm:epsrsultposeasy}
and Theorem \ref{thm:epsresultnegeasy}.  One of the goals of this paper is to provide a framework for using capacity theory to show when these auxiliary polynomials do or do not exist.  We give an outline in \S \ref{s:fieldguide} of how one proves these types of results.  In the conclusion we summarize the implications of our results and discuss possible directions for future research.  

\section{Background and Related Work}


Given a polynomial $f(x) = x^d + f_{d-1} x^{d-1} + \cdots + f_1 x + f_0 \in \Z[x]$ 
and a prime $p$ we can find solutions $x \in \Z$ to the equation
\begin{equation}
	f(x) \equiv 0 \bmod p
\end{equation}
in randomized polynomial time using e.g. Berlekamp's algorithm or the
Cantor-Zassenhaus algorithm \cite{B67,CZ81}.  While it is ``easy'' to
find roots of $f(x)$ in the finite field $\ZZ{p}$ and over $\Z$ as
well, there is no known efficient method to find roots of $f(x)$
modulo $N$ for large composite integers $N$ unless one knows the
factorization of $N$.



\subsection{Coppersmith's method}
\label{s:copper}

Although finding roots of a univariate polynomial, $f(x)$, modulo $N$ is difficult in general, if $f(x)$ has a ``small'' root, 
then this root can be found efficiently using Coppersmith's method~\cite{C97}.

 

Coppersmith's method for proving Theorem \ref{thm:Cop1} works as follows.  We follow the exposition in~\cite{CoppersmithFinding}, which incorporates simplifications due to Howgrave-Graham~\cite{H97}.
Suppose $\epsilon > 0$ and that $f(x)$ has a root $r \in \Z$ with $|r| \le N^{(1/d) - \epsilon} $
and $f(r) \equiv 0$ mod $N$. 
He considers the finite rank lattice $\mathcal{L}$ of rational polynomials in $\mathbb{Q}[x]$ of the form
\[
	h_{ij}(x) = \sum_{0 \le i+dj < t} a_{i,j} \ x^i \ (f(x)/N)^{j}
\]
where $t \ge 0$ is an integer parameter to be varied and all $a_{i,j} \in \Z$.  Here $\mathcal{L}$ is
a finite rank lattice because the denominators of the coefficients of $h_{ij}(x)$ are bounded
and $h_{ij}(x)$ has degree bounded by $t$.  

If we evaluate any polynomial $h_{ij} \in \mathcal{L}$ at a root $r$ satisfying $f(r) \equiv 0 \bmod N$, $h_{ij}(r)$ will be an integer.

Concretely, one picks a basis for a sublattice of $\mathcal{L} \in \Q^{t-1}$ 
by taking a suitable set of polynomials $\{h_{ij}(x)\}_{i,j}$ and representing each polynomial by its coefficient vector.
Coppersmith's method applies the LLL algorithm to this sublattice basis to find a short vector representing a specific polynomial, $h_\epsilon(x)$ 
in $\mathcal{L}$.  He shows that the fact that the vector
of coefficients representing $h_\epsilon(x)$ is short implies that $|h_{\epsilon}(x)| < 1$ for all $x \in \mathbb{C}$ with $|x| \le  N^{(1/d) - \epsilon}$, and that for sufficiently large $t$, the LLL algorithm will find a short enough vector.
Because $h_\epsilon(x)$
is an integral combination of terms of the form $x^i \ (f(x)/N)^j$, this forces
$h(r) \in \Z$ because $f(r)/N \in \Z$.  But $|r| \le N^{(1/d) - \epsilon}$ forces $|h_\epsilon(r)| < 1$.
Because $0$ is the only integer less than $1$ in absolute value, we see $h_\epsilon(r) = 0$.
So $r$ is among the zeros of $h_\epsilon(x)$, and as discussed earlier, there is an efficient method to find the integer zeros of a polynomial in $\mathbb{Q}[x]$.  One then lets $\epsilon \to 0$ and 
does a careful analysis of the computational complexity of this method.  

The bound in Theorem~\ref{thm:Cop1} arises from cleverly choosing a subset of the possible $\{h_{ij}\}$ as a lattice basis so that one can bound the determinant of the lattice as tightly as possible, then using the LLL algorithm in a black-box way on the resulting lattice basis.

\subsection{Optimality of Coppersmith's Theorem}

Since Coppersmith's technique uses the LLL algorithm \cite{lll} to
find the specific polynomial $h(x)$ in the lattice $\mathcal{L}$, it
is natural to think that improvements in lattice reduction techniques
or improved bounds on the length of the shortest vector in certain
lattices might improve the bound $N^{1/d}$ in Theorem~\ref{thm:Cop1}.

Such an improvement would be impossible in polynomial time for
arbitrary $N$, since the polynomial $f(x) = x^d$ has exponentially
many roots modulo $N=p^d$ of absolute value $N^{1/d+\epsilon}$, but
this does not rule out the possibility of improvements for cases of
cryptographic interest, such as polynomial congruences modulo RSA moduli $N = pq$.

Aono, Agrawal, Satoh, and Watanabe~\cite{aono} showed that
Coppersmith's lattice basis construction is optimal under the
heuristic assumption that the lattice behaves as a random lattice;
however they left open whether improved lattice bounds or a
non-lattice-based approach to solving this problem could improve the
$N^{1/d}$ bound.

\subsection{Cryptanalytic Applications of Coppersmith's Theorem}
\label{sec:applications}

Theorem \ref{thm:Cop1} has many immediate applications to
cryptanalysis, particularly the cryptanalysis of RSA. 
May~\cite{May} gives a comprehensive survey of cryptanalytic
applications of Coppersmith's method.  In this paper, we focus on
Coppersmith's method applied to univariate polynomials modulo
integers.  We highlight several applications of the univariate case below.

 The RSA
assumption posits that it is computationally infeasible to invert the
map $x \mapsto x^d \bmod N$, i.e., it is infeasible to find roots of
$f(x) = x^d - c \bmod N$.  Because of their similar structure, almost
all of the cryptographically hard problems (some of which are outlined
below) based on factoring can be approached using Coppersmith's method
(Theorem \ref{thm:Cop1}).


\paragraph{Low public exponent RSA with stereotyped messages:}

A classic example listed in Coppersmith's original paper \cite{C97} is decrypting ``stereotyped'' messages encrypted under low public exponent RSA, 
where an approximation to the solution is known in advance.  The general RSA map is $x \mapsto x^e \bmod N$. For efficiency purposes,
$e$ can be chosen to be as small as $3$, so that a ``ciphertext'' is $c_0 = x_0^3 \bmod N$.
Suppose we know some approximation to the message $\tilde{x}_0$ to the message $x_0$.
Then we can set
\[
	f(x) = (\tilde{x}_0 + x)^3 - c.
\]
Thus $f(x)$ has a root (modulo $N$) at $x = x_0 - \tilde{x}_0$.
If $|x_0 - \tilde{x}_0| < N^{(1/3)}$ then this root can be found using Coppersmith's method.

\paragraph{Security of RSA-OAEP}

The RSA function $x \mapsto x^e \bmod N$ is assumed to be a one-way trapdoor permutation.
Optimal Asymmetric Encryption Padding (OAEP) is a general method for taking a one-way trapdoor permutation and a random oracle \cite{BR93}, and creating 
a cryptosystem that achieves security against adaptive chosen ciphertext attacks (IND-CCA security).

Instantiating the OAEP protocol with the RSA one-way function yields RSA-OAEP -- a standard cryptosystem.
When the public exponent is $e = 3$, Shoup used Coppersmith's method to show that RSA-OAEP is secure against an adaptive chosen-ciphertext attack (in the random oracle model) \cite{S01}.

Roughly, the proof of security works as follows.  Suppose there is a distinguisher for the IND-CCA security of the RSA-OAEP cryptosystem.
Then, given a $y \equiv x^3 \bmod N$, there is an adversary that can use this distinguisher to extract an approximation $\tilde{x}$ for $x$.
Then, using Coppersmith's method, the adversary can recover $x$.  Since the RSA function is assumed to be one-way, there can be no such adversary, 
thus there can be no such distinguisher.  
Unfortunately, since Coppersmith's method requires $|\tilde{x} - x| < N^{1/e}$, this proof only goes through when $e$ is small.

\paragraph{Hard-core bits of the RSA Function}

Repeated iteration of the RSA function has been proposed as candidate for a pseudo random generator.
In particular, we can create a stream of pseudo random bits by picking an initial ``seed'', $x_0$ and calculating the series
\begin{align*}
	x_i &\mapsto x_{i+1} \\
	x_i &\mapsto x_i^e \bmod N
\end{align*}
At each iteration, the generator will output the $r$ least significant bits of $x_i$.
For efficiency reasons, we would like $r$ to be as large as possible while still maintaining the provable security of the generator.

When we output only $1$ bit per iteration, this was shown to be secure \cite{ACGS88,FS00}, and later this was increased to allow the generator 
to output any $\log \log(N)$ consecutive bits \cite{HN04}.  The maximum number of bits that can be safely outputted by such a generator is tightly tied to the approximation 
$\tilde{x}$ necessary for recovering $x$ from $x^e \bmod N$.
Thus a bound on our ability to find small roots of $f(x) = (x-\tilde{x})^e -c \bmod N$ immediately translates into bounds on the maximum number of bits 
that can be safely outputted at each step of the RSA pseudo random generator.

In order to construct a provably secure pseudo random generator that 
outputs $\Omega(n)$ pseudo random bits for each multiplication modulo $N$
\cite{SPW06a} assume there is no probabilistic polynomial time algorithm for solving 
the $\inparen{ \frac{1}{d} + \epsilon, d }$-SSRSA problem.

\begin{definition}[The $(\delta,d)$-SSRSA Problem \cite{SPW06a}]
	Given a random $n$ bit RSA modulus, $N$ and a polynomial 
	$f(x) \in \Z[x]$ with $\degree(f) = d$, 
	find a root $x_0$ such that $|x_0| < N^{\delta}$.
\end{definition}

Coppersmith's method solves the $\inparen{ \frac{1}{d}, d }$-SSRSA
Problem.  Our results show that Coppersmith's method cannot be used to
solve the $\inparen{ \frac{1}{d} + \epsilon, d }$-SSRSA problem.  Note
that our results do not show the $\inparen{
  \frac{1}{d}+\epsilon,d}$-SSRSA problem is intractable---doing so
would imply there is no polynomial-time algorithm for factoring---but
instead we show that the best available class of techniques cannot be
extended.

\paragraph{Finding Small Roots of Polynomial Equations in Theoretic Cryptography}
\label{sec:menagerie}

Coppersmith's method has led to a number of successful cryptanalytic attacks on the RSA cryptosystem (outlined above), 
but Coppersmith's method can also be applied to a variety of other factoring based cryptosystems that 
are common in the cryptographic literature.
Indeed, Coppersmith's method provides the best known approaches to breaking these cryptosystems, 
but their security has not been as extensively investigated as that of RSA because they are not as widely deployed.

Below, we list a number of other well-known cryptosystems that are susceptible to similar attacks to those described above.
Although the problems listed below are all still believed to be hard in general, setting parameters to ensure security requires 
a rigorous understanding of the best attacks, and their limitations.


The \emph{Quadratic Residuosity} (QR) assumption \cite{GM84} assumes that the set of quadratic residues 
modulo $N$ is indistinguishable from the set of elements with Jacobi symbol 1.
This assumption essentially states that not only is it infeasible to find roots of $f(x) = x^2 -a \bmod N$, 
it is infeasible to decide whether such an equation even has a root.

The Benaloh and Naccache-Stern cryptosystems \cite{Benaloh:94,NaccacheStern:98} extend the QR assumption to higher degree polynomials.
Specifically, they assume that if $N = pq$, and there is an $r > 1$ with $r | p - 1$, 
then there is no polynomial time algorithm that can distinguish the uniform distribution on the subgroup of $r$th powers modulo $N$ from the uniform distribution on the 
entire multiplicative group.  Given a value $a \in \ZZstar{N}$, then the polynomial $f(x) = x^r - a$ has a root modulo $N$ 
if and only if $a$ is an $r$th power.  Thus finding an algorithm that could find a root of $f$ would lead to a distinguisher.

The Okamato-Uchiyama cryptosystem \cite{OU98} works with moduli of the form $N = p^2q$, and the authors assume that 
there is no polynomial time algorithm that can distinguish random samples from the set of $p$th powers modulo $N$ 
from uniform elements in the multiplicative group modulo $N$.  Since $p | N$, any algorithm that could find roots 
of $f(x) = x^N - a \bmod N$ would break the security of this cryptosystem.

The security of the Paillier Cryptosystem \cite{P99} rests on the \emph{Decisional Composite Residuosity} assumption (DCR), 
which is the assumption that there is no polynomial time algorithm that can distinguish the uniform distribution on $N$th powers in $\ZZ{N}$ from uniform elements in $\ZZstar{N}$.
Thus any algorithm that finds roots of $f(x) = x^N - a \bmod N^2$ would break the security of this cryptosystem.
Because of its homomorphic properties, the Paillier Cryptosystem is a building block for many cryptographic protocols 
e.g. private searching on streaming data \cite{OS07} and private information retrieval \cite{Chang04,Lipmaa05}.

\subsubsection{Extensions to Coppersmith's method}

Coppersmith's original work also considered the problem of finding
small solutions to polynomial equations in two variables over the
integers and applied his results to the problem of factoring RSA
moduli $N=pq$ when half of the most or least significant bits of one
of the factors $p$ is known.~\cite{C97} Howgrave-Graham gave
an alternate formulation of this problem by finding approximate common
divisors of integers using similar lattice-based techniques, and
obtained the same bounds for factoring with partial
information.~\cite{HG01} May~\cite{M10} gives a unified formulation of
Coppersmith and Howgrave-Graham's results to find small solutions to
polynomial equations modulo unknown divisors of integers.  Later work
by Jutla~\cite{jutla} and Jochemsz and May~\cite{jochemsz} has
generalized Coppersmith's method to multivariate equations, and
Herrmann and May~\cite{Herrmann2008} obtained results for multivariate
equations modulo divisors.

As we will show in the next section, existing results in capacity
theory can be used to directly address the case of auxiliary
polynomials for Coppersmith's method for univariate polynomials modulo
integers.  Adapting these results to the other settings of
Coppersmith's method listed above is a direction for future research.


\section{Capacity Theory for Cryptographers}
\label{sec:capacity}

In this section, we use arithmetic capacity theory to prove our main results.
Classically, capacity theory arose from the following problem in electrostatics.  How will
a unit charge distribute itself so as to minimize potential energy if it is constrained
to lie within a  compact subset
$E_\infty$ of $\mathbb{C}$ which is stable under complex conjugation?

It was discovered by Fekete and Szeg\H{o} \cite{Fekete,FS} that  the distribution of small
charges on such an $E$  is related to the possible locations of zeros of monic integral
polynomials.  Heuristically, these zeros behave in the same way as unit 
charges that repel one another according to an inverse power law.   This
is due to the restriction that the discriminant of a monic integral polynomial
without multiple zeros must be a non-zero integer, which prevents all the zeros from being
too close to one another.  More precisely, the total potential energy of the charges
behaves as the negative of the logarithm of the absolute value of the 
discriminant of the above polynomial.    Both are sums over distinct pairs of
charges (roots) at positions $w$ and $z$ of $-\ln|z-w|$.  Since the
discriminant of the polynomial has absolute value at least $1$, the potential energy is
not positive.   

This heuristic is behind the following striking result of Fekete and Szeg\H{o} from \cite{Fekete,FS}.  
Define the capacity $\gamma(E_\infty)$ to be $e^{-V(E_\infty)}$,
where $V(E_\infty)$ is the so-called Robbin's constant giving the minimal
potential energy of a unit charge distribution on $E_\infty$.  Fekete and Szeg\H{o} showed that
if $\gamma(E_\infty) < 1$, then there are only finitely many 
irreducible monic polynomials with integer coefficients which have all of their
roots in $E_\infty$.  This corresponds to the case in which the minimal potential energy
$V(E_\infty)$ is positive, consistent with the above heuristic.  Conversely, if 
$\gamma(E_\infty) > 1$, 
then for every open neighborhood $U$ of $E_\infty$ in $\mathbb{C}$, there
are infinitely many irreducible monic polynomials with integer coefficients having
all their roots in $U$.   

The work of Fekete and Szeg\H{o} was vastly generalized by Cantor \cite{Cantor}
to adelic subsets of the projective line, and by Rumely \cite{Rumely1,Rumely2}
to adelic subsets of arbitrary smooth projective curves over global fields.  Their
methods are based on potential theory, as in electrostatics.  In \cite{ChinburgCap},
Chinburg suggested a simpler approach, called sectional capacity theory, which
applies to arbitrary regular projective varieties of any dimension and not just to curves.
Sectional capacity theory was based on ideas from Arakelov theory, with the geometry of numbers 
and Minkowski's theorem being the primary tools.
In \cite{RumelyLauVarley}, Rumely, Lau and Varley showed that the limits hypothesized in
\cite{ChinburgCap} do exist under reasonable hypotheses;  this is a deep result.

This paper is the first application of capacity theory that we are aware of  
to cryptography.  We will show that  capacity theory is very suited to studying
the kind of auxiliary polynomials used in the proof of Coppersmith's theorem.
Before we begin, however, we review some  number theory.
\subsection{$p$-adic Numbers}
\label{s:prel}


For any prime $p$, and any $n \in \Z$, we define the \emph{$p$-adic} valuation of $n$, to be the supremum of the 
integers  $e$ such that $p^e | n$, \ie
\[
	v_p(n) = \left\{ \begin{array}{ll} \max \inbraces{ e \in \Z : p^e \mid n } & \mbox{ if $n \ne 0$ } \\ \infty & \mbox{ if $n = 0$ } \end{array} \right.
\]
This is then extended to rational numbers in the natural way.  If $a,b \in \Z$ and $a,b \ne 0$, then 
\[
	v_p \inparen{ \frac{a}{b} } = v_p(a) - v_p(b).
\]
The $p$-adic valuation gives rise to a $p$-adic absolute value $| \ |_p:\Q \to \mathbb{R}$ given by 
\begin{equation}
\label{eq:properties}
	\pnorm{x} = \left\{ \begin{array}{ll} p^{-v_p(x)} & \mbox{ if $x \ne 0$ }, \\ 0 & \mbox{ if $x = 0$ }. \end{array} \right.
\end{equation}
It is straightforward to check that the $p$-adic absolute value is multiplicative and satisfies a stronger form of the triangle inequality:
\begin{equation}
\label{eq:nonarch}
	\pnorm{xy} = \pnorm{x}\cdot \pnorm{y} \quad \mathrm{and} \quad \pnorm{x+y} \le \max \inparen{ \pnorm{x} , \pnorm{y} }\quad \mathrm{for} \quad x, y \in \Q.
\end{equation}
The $p$-adic absolute value defines a metric on  $\Q$.  The $p$-adic numbers, $\Qp$, are defined to be the completion 
of $\Q$ with respect to this metric.  This is similar to the construction of $\R$ as the completion of $\Q$ with respect to the Euclidean absolute value $|\ |:\Q \to \R$.

Elements of $\Qp$ are either $0$ or expressed in a unique way as a formal infinite sum
\[
	\sum_{i=k}^\infty a_i p^i 
\]
in which $k \in \Z$, each $a_i$ lies in $\{0,1,\ldots,p-1\}$ and $a_k \ne 0$.
Such a sum converges to an element of $\Qp$  because the sequence of integers $\{s_j\}_{j = k}^\infty$ defined by 
$s_j = \sum_{i = k}^j a_i p^i$ forms a Cauchy sequence with respect to the metric $| \ |_p$.  One can add, subtract and multiply such sums by treating $p$ as a formal variable, performing operations in the resulting formal power series ring in one variable over $\Z$, and by then carrying appropriately.  In fact, $\Qp$ is a field, since multiplication is commutative and it is possible to divide elements by non-zero elements of $\Qp$.

A field $L$ is algebraically closed if every non-constant polynomial $g(x) \in L[x]$
has a root in $L$.  This implies that $g(x)$ factors into a product of linear polynomials
in $L[x]$, since one can find in $L$ roots  of quotients of $g(x)$ by products
of previously found linear factors.   For example, $\mathbb{C}$ is algebraically
closed, but $\Q$ is certainly not.  

In general, given a field $F$ there
are many algebraically closed fields $L$ containing $F$.  For example, given one such $L$,
one could simply label the elements of $L$ by the elements of some other set, or one could put $L$ inside a larger algebraically closed field.  Given one $L$,
the set $\overline{F}$ of elements $\alpha \in L$ which 
are roots in $L$ of some polynomial in $F[x]$ is called the algebraic closure of 
$F$ in $L$.  The set $\overline{F}$ is in fact an algebraically closed field.  For a given
$F$, the algebraic closure $\overline{F}$ will depend on the algebraically closed field $L$ which one chooses in this construction.  But if one were to use a different field $\tilde{L}$, say, then the
algebraic closure of $F$ in $\tilde{L}$ is isomorphic to $\overline{F}$ by a (non-unique)
isomorphism which is the identity on $F$.  So we often just fix one algebraic closure
$\overline{F}$ of $F$.

For instance, if $F = \mathbb{Q}$, then $L = \mathbb{C}$ is algebraically closed,
so we can take $\overline{\mathbb{Q}}$ to be the algebraic closure of $\mathbb{Q}$
in $\mathbb{C}$. The possible field embeddings
$\tau:\overline{\mathbb{Q}} \to L = \mathbb{C}$ come from pre-composing with a field automorphism
of $\overline{\mathbb{Q}}$.

However, for each prime $p$, there is another alternative.
The field $\Qp$ is not algebraically closed, but as noted above, we can find
an algebraically closed field containing it and then construct the algebraic
closure $\Qpbar$ of $\Qp$ inside this field.  Now we
have $\mathbb{Q} \subset \mathbb{Q}_p \subset \Qpbar$, and
$\Qpbar$ is algebraically closed.  So we could take $L = \Qpbar$ and consider the
algebraic closure $\overline{\mathbb{Q}}'$of $\mathbb{Q}$ inside $\Qpbar$. 
We noted above that all algebraic closures of $\mathbb{Q}$ are isomorphic over $\mathbb{Q}$
in many ways.  The possible isomorphisms of $\overline{\mathbb{Q}}$ (as a subfield of
$\mathbb{C}$, for example) with $\overline{\mathbb{Q}}'$ (as a subfield
of $\Qpbar$) correspond to the field embeddings $\sigma: \overline{\mathbb{Q}}  
\to \Qpbar$.  Each such $\sigma$ gives an isomorphism of $\overline{\mathbb{Q}}$
with $\overline{\mathbb{Q}}'$ which is the identity map on $\mathbb{Q}$.
Note here that $\Qpbar$ is much larger than $\overline{\mathbb{Q}}$, since $\Qpbar$
(and in fact $\Qp$ as well) is uncountable while $\Qbar$ is countable.

Each $\alpha \in \Qbar$ is a root of a unique monic polynomial $m_\alpha(x) \in \Q[x]$ of minimal degree, and $m_\alpha(x)$ is irreducible. We will later need to discuss
the image of such an $\alpha$ under all the field embeddings $\tau:\Qbar \to \mathbb{C}$
and under all field embeddings $\sigma:\Qbar \to \Qpbar$ as $p$ varies.  
The possible values for $\tau(\alpha)$ and $\sigma(\alpha)$ are simply
the different roots of $m_\alpha(x)$ in $\mathbb{C}$ and $\Qpbar$, respectively.

\begin{example}
If $\alpha = \sqrt{7}$ then $m_\alpha(x) = x^2 - 7$. The 
possibilities for $\tau(\alpha)$ are the positive real square root
$2.64575...$ and the negative real square root
$-2.64575...$ of $7$.  When $p = 3$, it turns out that $x^2-7$ already has two roots
$\alpha_1$ and $\alpha_2$ in the $3$-adic numbers $\Q_3 \subset \Qbar_3$.  These roots are
$$\alpha_1 = 1 + 1\cdot 3 + 1\cdot 3^2 + 0\cdot 3^3 + \cdots\quad \mathrm{and}\quad 
\alpha_2 = 2 + 1 \cdot 3 + 1 \cdot 3^2 + 2 \cdot 3^3 + \cdots.$$
These expansions result from choosing $3$-adic digits so that the square
of the right hand side of each equality is congruent to $1$ modulo an increasing
power of $3$.  This is the $3$-adic counterpart of finding the decimal digits of the
two real square roots of $7$.  So the possibilities for $\sigma(\alpha)$ under all embeddings $\sigma:\Qbar \to 
\Qbar_3$ are $\alpha_1$ and $\alpha_2$.
\end{example}

Basic facts about integrality and divisibility are naturally encoded using $p$-adic absolute values:

\begin{fact}
\label{f:absextend} As above, let $\Qpbar$ denote an algebraic closure of $\Qp$.
There is a unique extension of $| \ |_p:\Qp \to \mathbb{R}$ to an absolute value $| \ |_p:\Qpbar \to \mathbb{R}$ for which (\ref{eq:nonarch}) holds for all $x, y \in \Qpbar$.  
\end{fact}

\begin{fact}
\label{f:algint}  The set $\overline{\Z}$ of algebraic integers is the set of all $\alpha \in \Qbar$ for which $m_\alpha(x) \in \Z[x]$.  In fact, 
$\overline{\Z}$ is a ring, so that adding, subtracting
and multiplying algebraic integers produces algebraic integers.  One can 
speak of congruences in $\overline{\Z}$ by saying $\alpha \equiv \beta$ 
mod $\gamma \overline{\Z} $ if $\alpha - \beta = \gamma \cdot \delta$
for some $\delta \in \overline{\Z}$.
\end{fact}

\begin{fact}
\label{f:absbound}
	If $r \in \Q$ then $\pnorm{r} \le 1$ for all primes $p$ if and only if $r \in \Z$.  More generally, an element $\alpha \in \Qbar$ is in $\overline{\Z}$
	if and only if for all primes $p$ and all field embeddings $\sigma:\Qbar \to \Qpbar$ one has $|\sigma(\alpha)|_p \le 1$.
\end{fact}

\begin{fact}
\label{f:zero}  Suppose $\alpha \in \overline{\Z}$ and $|\tau(\alpha)| < 1$
for all embeddings $\tau:\Qbar \to \mathbb{C}$.  Then 
in fact, $\alpha = 0$.  To see why, note that $m_\alpha(0) \in \Z$ is  
$\pm 1$ times the product of the complex roots of $m_\alpha(x)$.  These roots all have the form
 $\tau(\alpha)$, so $|m_\alpha(0)| < 1$. Then $m_\alpha(0) \in \Z$
 forces $m_\alpha(0) = 0$.  Because $m_\alpha(x)$ is monic and irreducible this means $m_\alpha(x) = x$, 
so $\alpha = 0$.
\end{fact}

\begin{fact}
	If $N = pq$ for distinct primes $p$ and $q$, then $\pnorm{N} = \frac{1}{p}$, $\ppnorm{N}{q} = \frac{1}{q}$, and $\ppnorm{N}{p'} = 1$ for all other primes $p'$.
\end{fact}

\begin{fact}
  \label{fact:div}
	If $a,b \in \Z$, then 
	\[
		a | b \qquad \Leftrightarrow \qquad \pnorm{b} \le \pnorm{a} \quad \forall p
	\]
	Thus $a | b$ is the statement that $b$ is in the $p$-adic disc of radius $|a|_p$ centered at 0 for all $p$.
	More generally, if $\alpha, \beta \in \overline{\Z}$ then $\alpha$ divides $\beta$
	in $\overline{\Z}$ if $\beta = \delta \cdot \alpha $ for some $\delta \in \overline{\Z}$. This is so if and only if $|\sigma(\beta))|_p \le |\sigma(\alpha))|_p$ for all primes $p$ and all field embeddings $\sigma: \Qbar \to \Qpbar$.  
\end{fact}

\subsection{Auxiliary Functions}
\label{s:auxf}

The original question Coppersmith considered was this: Given an integer $N \ge 1$, a polynomial $f(x)$, and a bound $X$, can we find 
all integers $z \in \Z$ such that $|z| \le X$ and  $f(z) \equiv 0 \bmod N$?  

When $X$ is sufficiently small in comparison to $N$, Coppersmith constructed a non-zero auxiliary polynomial of the form
\begin{equation}\label{eqn:coppform}
h(x)=\sum_{i,j} a_{i,j}x^i(f(x)/N)^j,\qquad a_{i,j}\in\Z
\end{equation}
satisfying $|h(z)|<1$ for every $z\in\C$ with $|z|\le X$.  As noted in Section~\ref{s:copper}, this boundedness property forces the set of $z\in\Z$ satisfying
$|z|\le X$ and $f(z)\equiv 0\bmod N$ to be among the roots of $h(x)$.  In fact,
the roots of the $h(x)$ include all
algebraic integers $z \in \Zbar$ satisfying
\begin{equation}
\label{eq:niceprops}
 f(z) \equiv 0 \bmod N\cdot \Zbar\quad \mathrm{and}\quad |\sigma(z)| \le X\quad \mathrm{for \ all \ embeddings}\quad  \sigma : \Qbar \to \C.
\end{equation}
The reason is as follows.  For $z\in\Zbar$, the condition that $f(z)\equiv 0\mod N\Zbar$ is equivalent to the condition that $f(z)/N\in\Zbar$. Therefore, for any $h(x)$ in the form of Equation~\ref{eqn:coppform}, we have $h(z)\in\Zbar$
whenever $f(z)\equiv 0\bmod{N\Zbar}$.  
If $h(x)$ further satisfies $|h(z)|<1$ for all $z\in\C$ with $|z|\le X$, then the property that
$|\sigma(z)|\le X$
for all embeddings $\sigma:\Qbar\to\C$, means that $|h(\sigma(z))|<1$ as well. Fact~\ref{f:zero} therefore tells us that $h(z)=0$. 

Capacity theory can be used for solving the problem of deciding whether there exist non-zero auxiliary polynomials
$h(x)$ which include among its roots the set of  $z\in\Zbar$ satisfying Equation~\ref{eq:niceprops}.  The basic idea, which will be given in detail in Section~\ref{sec:mainresult}, is
that capacity theory gives one a way of deciding whether the set of algebraic integers satisfying Equation~\ref{eq:niceprops} is finite or infinite.  

When this set is infinite then 
there cannot exist {\it any} rational function $h(x)$ of any kind vanishing on the $z\in\Zbar$ satisfying (\ref{eq:niceprops}), and in particular no $h(x)$ of the form in (\ref{eqn:coppform})
will exist satisfying the desired properties.  If, on the other hand, this set is finite then there {\it will} exist an auxiliary polynomial $h(x)$ vanishing on the $z\in\Zbar$ satisfying (\ref{eq:niceprops}),
and in fact Coppersmith explicitly constructed such a polynomial using the LLL algorithm.  As we will see, the boundary for finite versus infinite occurs when $X=N^{1/d}$ where $d$ is the degree of $f(x)$.

\subsection{When do Useful Auxiliary Polynomials Exist?}
\label{sec:mainresult}

In this section, we use capacity theory to give a characterization of when auxiliary polynomials $h(x)$
of the kind discussed in \S \ref{s:auxf} exist. We will use the work of Cantor in \cite{Cantor}
to show the following result.

\begin{theorem}[Existence of an auxiliary polynomial]
		\label{thm:exist}
		\label{thm:main}
			Let $d$ be the degree of $f(x)$.  
	Define  $S(X)$ to be the set of all 
			all algebraic integers $z \in \Zbar$ such that 
			\[
				f(z) = 0 \bmod N \Zbar \quad \mathrm{and}\quad |\sigma(z)| \le X\quad \mathrm{for \ all \ embeddings}\quad \sigma:\Qbar \to \C	.	\]			
		There exists a polynomial $h(x) \in \Q[x]$ whose roots include
			every element of $S(X)$ if 
			$X < N^{1/d}$. If $X > N^{1/d}$
			there is no rational function $h(x) \in \Q(x)$ whose zero set contains
			 $S(X)$ because $S(X)$ is infinite.
	\end{theorem}
	
We break the proof into a sequence of steps.

\begin{enumerate}

	\item
		Since $f(x) \in \Z[x]$, and embeddings fix integers, then if $z \in \Zbar$ we have $f(z) \in \Zbar$, and 
		$\sigma(f(x)) = f(\sigma(x))$ for all embeddings $\sigma: \Qbar \to \Qpbar$.
	\item
		Suppose $N = p_1^{e_1} \cdots p_k^{e_k}$ and $x \in \Z$, 
		then by Fact~\ref{fact:div}
		\begin{align*}
			f(z) \equiv 0 \bmod N &\Leftrightarrow \ppnorm{f(z)}{p_i} \le \inparen{\frac{1}{p_i}}^{e_i} \quad \forall i \in [k] \\
							&\Leftrightarrow \ppnorm{f(z)}{p_i} \le \ppnorm{N}{p_i} \quad \forall i \in [k]
		\end{align*}
		Similarly, if $z \in \Zbar$ then 
		\begin{equation*}
			f(z) = 0 \bmod N \Zbar
							\Leftrightarrow
							|\sigma(f(z))|_{p_i} = \ppnorm{f(\sigma(z))}{p_i} \le \ppnorm{N}{p_i}
		\end{equation*}
                                                       for all $i \in [k]$ and for all embeddings $\sigma:\Qbar \to \Qpbar$.
	      \item
		\label{itm:Ep}
		For all primes, $p$, define the set of elements in $\Qpbar$ that solve the congruence in Equation~\ref{eq:niceprops} $p$-adically:
		\[
			E_p \defined \inbraces{ z \in \Qpbar \suchthat \pnorm{f(z)} \le \pnorm{N} } = f^{-1}\inparen{ \inbraces{ z \in \Qpbar \suchthat \pnorm{z} \le \pnorm{N} } },
		\]

		and similarly define the set of elements with bounded complex absolute value
		\[
			E_\infty \defined \inbraces{ z \in \C \suchthat |z| \le X }
		\]
		Let 
		\[
			\E \defined  E_\infty \times \prod_{p \in \primes} E_p
		        \]
                        This specifies the set of $p$-adic and complex constraints on our solutions.  We will compute the capacity of $\E$, a measurement of the size of $\E$. 
	\item
		We now define the \emph{local capacities} $\gamma_p(E_p)$ and $\gamma_\infty(E_\infty)$ as well as the global capacity $\gamma(\E)$.  
		Suppose $0 \le r \in \mathbb{R}$.  We have $p$-adic and complex discs of radius $r$
		defined by
				\[
			D_p(a,r) = \inbraces{ z \in \Qpbar \suchthat \pnorm{z-a} \le r }\quad \mathrm{for} \quad a \in \Qpbar
		\]
		and
		\[
			D_\infty(a,r) = \inbraces{ z \in \C \suchthat |z - a| \le r }\quad \mathrm{for}\quad a \in \C.
		\]
		\begin{fact}[Capacity of a Disc]
		\label{fact:capdisc}
		For $v = p$ and $v = \infty$, one has local capacity
		\[
			\gamma_v ( D_v(a,r) ) = r
		\]
	    If $v = p$, $a = 0$ and $r = |N|_p$ is the $p$-adic absolute value of an integer $N \ge 1$,
	    then $D_v(0,|N|_p) \cap \mathbb{Z}_p$ is just $N \mathbb{Z}_p$.  We will need later the
	    fact that the $p$-adic capacity of $N\mathbb{Z}_p$ is 
	    $$\gamma_p(N\mathbb{Z}_p) = p^{-1/(p-1)} |N|_p$$
	    In a similar way, suppose $v = \infty$.  The capacity of the real interval $[-r,r]$ is
	    $$\gamma_\infty([-r,r]) = r/2$$
		\end{fact}
		\begin{fact}[Capacity of polynomial preimage]
			\label{fact:cappreimage}
			If $f(x) \in \Z[x]$ is a monic degree $d$ polynomial, and  $S$ is a subset of $\Qpbar$ if $v = p$ or of $\C$ if $v  = \infty$ for which the capacity $\gamma_v(S)$
			is well defined, then $\gamma_v(f^{-1}(S))$ is well defined and 
			\[
				\gamma_v \inparen{ f^{-1}(S) } = \gamma_v(S)^{1/d}
			\]
		\end{fact}

		Facts \ref{fact:capdisc} and \ref{fact:cappreimage} show that
		\[
			\gamma_p( E_p ) = \gamma_p( D_p(0,\pnorm{N}) )^{1/d} = |N|_p^{1/d}
		\]
		and
		\[
			\gamma_\infty( E_\infty ) = \gamma_p( D_\infty(0, X ) )= X
		\]
		
		\begin{fact}[Capacity of a product]
			\label{fact:capprod}
			If 
			\[
				S \defined S_\infty \times \prod_{p \in \primes} S_p
			\]
			then 
			\[
				\gamma(S) = \gamma_\infty(S_\infty) \cdot \prod_{p \in \primes} \gamma_p(S_p)
			\]
		\end{fact}
		Fact \ref{fact:capprod} means
		\begin{align*}
			\gamma(\E) 	&= \gamma_\infty(E_\infty) \cdot \prod_{p \in \primes} \gamma_p(E_p) \\
						&= X \cdot \prod_{p \in \primes} |N|_p^{1/d} \\
						&= X \cdot \prod_{i = 1}^k p_i^{-e_i/d}\\
						&= X \cdot N^{-1/d}
		\end{align*}
	\item Computing the capacity of our sets of interest tells us whether there exists a polynomial mapping the components of $\E$ into discs of radius 1.  This allows us to apply the following theorem, due to Cantor~\cite{Cantor}, which tells us when an auxiliary polynomial exists.
		\begin{theorem}[Existence of an auxiliary polynomial]
			\label{fact:aux}
			If 
			\[
				\E = E_\infty \times \prod_{p \in \primes} E_p
			\]
			then there exists a non-zero auxiliary polynomial $h(x) \in \Q[x]$ satisfying 
			\[
				h( E_p ) \subset D_p(0,1) \quad \forall p
			\]
			and
			\[
				h( E_\infty) \subset \inbraces{ z \in \C \suchthat |z| < 1 }
			\]
			if $\gamma(\E) < 1$, and no such polynomial exists if $\gamma(\E) > 1$.
		\end{theorem}
	\end{enumerate}
Once we have set up this framework, we are now ready to prove Theorem~\ref{thm:exist}.

\begin{proof}[Proof of Theorem \ref{thm:exist}] 
			Suppose first that $X < N^{1/d}$.  Then by Fact \ref{fact:capprod}, $\gamma(\E) < 1$.
			By Fact \ref{fact:aux}, there exists a polynomial $h(x) \in \Q[x]$  
			with $|h(z)|_p \le 1$ for all $p$ and $z \in E_p$, and $|h(z)| < 1$ for all $z \in E_\infty$.
		     Suppose $z \in S(X)$.  Then $f(z)/N \in \Zbar$, so
		     Fact \ref{f:absbound} says that 
		     for all primes $p$ and embeddings $\sigma:\Qbar \to \Qpbar$
		     one has
		     $$\pnorm{\sigma(f(z)/N)} \le 1$$
		     Since $f(x) \in \Z[x]$ and $N \in \Z$, we have $\sigma(f(z)) = f(\sigma(z))$ and $\sigma(N) = N$.  So 
		     $$\pnorm{f(\sigma(z))}  = \pnorm{\sigma(f(z))} = \pnorm{\frac{\sigma(f(z))}{\sigma(N)}} \cdot \pnorm{\sigma(N)} = \pnorm{\sigma(f(z)/N)} \cdot \pnorm{N} \le \pnorm{N}.$$
		     Therefore $\sigma(z) \in E_p$.  Hence $|h(\sigma(z))|_p \le 1$, where
		     $\sigma(h(z)) = h(\sigma(z))$ since $h(x) \in \Q[x]$. Because 
		     $p$ was an arbitrary prime, this means $h(z)$ is an algebraic integer, i.e.
		     $h(z) \in \Zbar$ by Fact \ref{f:absbound}.   On the other hand,
		     $z \in S(X)$ implies $|\sigma(z)| \le  X$ so $|\sigma(h(z))| = 
		     |h(\sigma(z))| < 1$ for all $\sigma: \Qbar \to \C$.  Thus
		     $h(z)$ is an algebraic integer such that $|\sigma(h(z))| < 1$ for
		     all $\sigma: \Qbar \to \C$, so by Fact \ref{f:zero}, $h(z) = 0$
		     as claimed.   When $X > N^{1/d}$, $S(X)$ is infinite by
		     \cite[Thm 5.1.1]{Cantor}.
\end{proof}

	To try to prove stronger results about small solutions of congruences, Coppersmith also considered auxiliary polynomials which absolute value less than $1$ on a real
	interval which is symmetric about $0$.  We can quantify his observation that this does not
	lead to an improvement of the exponent $1/d$ in Theorem \ref{thm:Cop1} by the following result.
	
	\begin{theorem}
	\label{thm:existreal}
			 Let $S'(X)$ be the subset of all $z \in S(X)$ such that $\sigma(z)$ lies in $\mathbb{R}$
			for every embedding $\sigma:\Qbar \to \C$.  
			There exists a polynomial $h(x) \in \Q[x]$ whose roots include
			every element of  $S'(X)$ if 
			 $X <  2 N^{1/d}$. If  $X > 2 N^{1/d}$
			there is no non-zero rational function $h(x) \in \Q(x)$ whose zero set contains
			 $S'(X)$ because  $S'(X)$ is infinite.

		\end{theorem}

	\begin{proof}[Proof of Theorem~\ref{thm:existreal}]
		     To prove the Theorem \ref{thm:existreal}, one just replaces the complex disc $E_\infty = 
		     \{z \in \mathbb{C}: |z| \le X\}$ by the real interval $E'_\infty = 
		     \{z \in \mathbb{R}:|z| \le X\}$.  Letting $\E' = \prod_p E_p \times E'_\infty$,
		     we find $\gamma(\E') = 	2 \cdot 	\gamma(\E)$ because
		     $\gamma(E'_\infty) = 2 \gamma(E_\infty)$.  So $\gamma(\E') < 1$
		     if $X < 2 N^{1/d}$ and we find as above that there is a
		     polynomial $h(x) \in \mathbb{Q}[x]$ whose roots contain every
		     element of $S(X)'$.  If $X > 2 N^{1/d}$ then $\gamma(\E') > 1$
		     and $S(X)'$ is infinite by the main result of \cite{Rumely2}, so $h(x)$
		     cannot exist.  		     
\end{proof}



\section{Lattices of Binomial Polynomials}
\label{s:improve}

In this section, we will answer the question of whether Coppersmith's theorem can be improved using
auxiliary polynomials that are combinations of binomial polynomials.
The results we proved in \S \ref{sec:capacity} showed that it is impossible to improve the bounds for auxiliary polynomials of the form
	$h(x) = \sum_{i,j \ge 0} a_{i,j} x^i (f(x)/N)^j$.

Recall that if $i \ge 0$
is an integer, the binomial polynomial $b_i(x)$ is
\[
b_i(x) = x \cdot (x-1) \cdots (x - i + 1)/i!.
\]

Based on a suggestion by Howgrave-Graham and Lenstra, Coppersmith considered in~\cite{CoppersmithFinding} auxiliary polynomials constructed from binomial polynomials; that is, of the form
\begin{equation}
	\label{eqn:binomiallattice}
	h(x) = \sum_{i,j \ge 0} a_{i,j} b_i(x) b_j(f(x)/N).
\end{equation}
He found that was he unable to improve the bound of $N^{1/d}$ using this alternate lattice.
In this section we will prove some sharper forms of Theorems
\ref{thm:epsrsultposeasy} and \ref{thm:epsresultnegeasy} that 
explain why this is the case.

Following the method laid out in \S \ref{sec:capacity}, we find that
capacity theory cannot rule out the existence of such polynomials.
One of the key differences is that monomials send algebraic integers
to algebraic integers, while binomial polynomials do not because of
the denominators.  Therefore, we are no longer able to use the same
sets $\E_p$ as in the previous section.

In fact, if one uses the lattice of binomial polynomials of the form (\ref{eqn:binomiallattice}), then for \emph{any} disk in $\C$
there \emph{do exist} auxiliary polynomials that have the required boundedness properties.  This is in contrast to the situation for polynomials
constructed from the monomial lattice.  In Theorem \ref{thm:epsrsultpos}, we
exhibit, for any disk, an explicit construction of such a polynomial.
However, since this polynomial is constructed with $j = 0$ in
\eqref{eqn:binomiallattice}, it tells us nothing about the solution to
the inputs to Coppersmith's theorem.

Theorem~\ref{thm:epsresultneg} shows that even if one manages to find
an auxiliary polynomial in the lattice given by
\eqref{eqn:binomiallattice} that does give nontrivial information
about the solutions to the inputs to Coppersmith's theorem, this polynomial will still not be useful.
Either this
polynomial must have degree so large that the root-finding step does
not run in polynomial time, or $N$ must have a small prime factor.  For
this reason, for $N$ that has only large prime factors,
using auxiliary polynomials constructed using binomial
polynomials will not lead to an improvement in the
$N^{1/d}$ bound in Coppersmith's method.

\begin{theorem}[Existence of bounded binomial polynomials]
\label{thm:epsrsultpos}Suppose $\delta$ is any positive real number.
\label{en:part1} Suppose $c > 1$.  For all sufficiently large integers $N$, there is a 
non-zero polynomial of the form  
\begin{equation}
\label{eq:binomialformpos}
h(x) = \sum_{0 \le i \le c N^{\delta}} a_{i} \ b_i(x) 
\end{equation}
with $a_i \in \mathbb{Z}$ such that $|h(z)| < 1$ for all 
$z$ in the complex disk $\{z \in \C : |z| \le N^{\delta} \}$. 
\end{theorem}

\begin{theorem}[Explicit construction for Theorem~\ref{thm:epsrsultpos}]
  \label{thm:epsrsultpos2}
Suppose $c > q_0 = 3.80572...$ when  
$q_0 $  is the  unique positive real number such that 
\begin{equation}
\label{eq:qbound}
  4 \mathrm{arctan}(q_0/2) = q_0\left(2\ln(2)  - \ln \left(\frac{4}{q_0^2} + 1\right)\right)\end{equation}
Then one can exhibit an explicit $h(x)$ of the kind in (\ref{en:part1}) in the following way. Choose any constant $c'$ with $q_0 < c' < c$.  Then for sufficiently
large $N$ and all integers $t$ in the range $c' N^\delta/2 < t \le c N^\delta/2 -1/2$, the function
$$h(x) = b_{2t+1}(x + t)$$
will have the properties in (i).
\end{theorem}

\begin{theorem}[Negative Coppersmith Theorem for binomial polynomials]
\label{thm:epsresultneg}  Suppose $\epsilon > 0$ and that $M $ and $N$ are positive integers. 
Suppose further that 
\begin{equation}
\label{eq:Nprod}
N^\epsilon > \prod_{p \le M} p^{1/(p-1)}
\end{equation}
where the product is over the primes $p$ less than or equal to $M$. 
This condition holds, for example, if $1.48774 N^\epsilon \ge M \ge 319$.  If there is 
a non-zero polynomial $h(x)$ of the form 
\begin{equation}
\label{eq:binomialformneg}
h(x) = \sum_{0 \le i,j \le M} a_{i,j} \ b_i(x) \ b_j(f(x)/N)
\end{equation}
with $a_{i,j} \in \mathbb{Z}$ such that $|h(z)| < 1$ for $z$ in the complex
disk  $\{z \in \C : |z| \le N^{(1/d)+ \epsilon} \}$, then $N$ must have 
a prime factor less than $M$.   
\end{theorem}

\subsection{Proof of Theorems \ref{thm:epsrsultpos} and \ref{thm:epsrsultpos2}.  }

The proof of Theorem \ref{thm:epsrsultpos} comes in several parts.  We
first use capacity theory to show that non-zero
polynomials of the desired kind exist.  This argument does not give any 
information about the degree of the polynomials, however.   So we then use an explicit geometry of numbers
argument to show the existence of a non-zero polynomial
of a certain bounded degree which is of the desired type.  Finally, we
give an explicit construction of an $h(x)$.  This $h(x)$ has a
somewhat larger degree than the degree which the geometry of numbers
argument shows can be achieved.  It would be interesting to see if the
LLL algorithm would lead to a polynomial time method for constructing a lower
degree polynomial than the explicit construction.

In this section we assume the notations of Theorem \ref{thm:epsrsultpos}.
The criterion that $h(x)$ be a polynomial of the form
$$h(x) = \sum_i a_i b_i(x)$$
with $a_i \in \mathbb{Z}$ is an \emph{extrinsic} property, which will be discussed in more detail in Step 1 of \S \ref{s:showexist}.  In short, this extrinsic property arises because $h(x)$ must have a particular form.  We need to convert this to an \emph{intrinsic} criterion, in this case observing that these polynomials take $\Z_p$ to $\Z_p$.  The key to doing so is the following
result of Polya:

\begin{theorem}
\label{thm:Polya}{\rm (Polya)}  The set of polynomials $h(x) \in  \mathbb{Q}[x]$ which have integral
values on every rational integer $r \in \mathbb{Z}$ is exactly the set of integral combinations
$\sum_i a_i b_i(x)$ of binomial polynomials $b_i(x)$.
\end{theorem}

\begin{corollary}
\label{cor:Polyacor}The set of polynomials $h(x) \in \mathbb{Q}[x]$ which are integral
combinations $\sum_i a_i b_i(x)$ of binomial polynomials $b_i(x)$ is exactly the set of $h(x)$
such that $|h(z)|_p \le 1$ for all $z \in \mathbb{Z}_p$ and all primes $p$.
\end{corollary}

\begin{proof} Suppose first that $|h(z)|_p \le 1$ for all $z \in \mathbb{Z}_p$ and all primes $p$.
Since each integer $z$ lies in $\mathbb{Z}_p$ for all primes $p$, we conclude that if $z$
is an integer then the rational number $h(z)$ lies in $\mathbb{Z}_p$ for all $p$.  Because
$\mathbb{Z} = \mathbb{Q} \cap  (\cap_p \mathbb{Z}_p)$ this shows $h(x)$ takes integers to integers.
Conversely, suppose $h(x) \in \mathbb{Q}[x]$ takes integers to integers and that $p$ is a prime.
Since $z \to h(z)$ defines a continuous function of $z \in \mathbb{Z}_p$ and $\mathbb{Z}$
is dense in $\mathbb{Z}_p$, we conclude that $h(z) \in \mathbb{Z}_p$ if $z \in \mathbb{Z}_p$.
\end{proof}

Our main goal in the proof of Theorem \ref{thm:epsrsultpos} is to show there are $h(x) \ne 0$ as in Corollary \ref{cor:Polyacor}
such that $|h(z)| < 1$ for $z$ in the complex disk $E_\infty = \{z \in \mathbb{C}: |z| \le N^\delta\}$.
\medbreak
We break reaching this goal into steps.
\medbreak
\noindent {\bf Applying capacity theory directly.}
\medbreak

In view of Corollary \ref{cor:Polyacor}, the natural adelic set to consider would be
\begin{equation}
\label{eq:natural}
\mathbb{E} = \prod_p E_p \times E_\infty \quad \mathrm{with} \quad E_p = \mathbb{Z}_p \quad \mathrm{for \ all } \quad p
\end{equation}
However, this choice does not meet the criteria for $\gamma(\mathbb{E})$ to be well defined,
because it is not true that $E_p = \overline{\mathbb{Z}}_p$ for all but finitely many $p$.  
However, for all  $Y \ge 2$, the adelic set 
\begin{equation}
\label{eq:natural}
\mathbb{E}' = \prod_{p \le Y} \mathbb{Z}_p \times \prod_{p > Y} \overline{\mathbb{Z}}_p \times E_\infty 
\end{equation}
does satisfy the criteria for $\gamma(\mathbb{E})$ to be well defined.  One has
$$\gamma_p(\mathbb{Z}_p) = p^{-1/(p-1)},\quad \gamma_p(\overline{\mathbb{Z}}_p) = 1 \quad \mathrm{and} \quad \gamma_\infty(E_\infty) = N^\delta.$$
So 
\begin{equation}
\label{eq:formal}
\ln \gamma(\mathbb{E}') = \ln \left( \prod_{p\le Y} \gamma_p(\mathbb{Z}_p) \times \gamma_\infty(E_\infty) \right) = -\sum_{p \le Y}  \frac{\ln(p)}{p-1} + \ln(N^\delta)
\end{equation}
Here as $Y \to \infty$, the quantity $-\sum_{p \le Y} \frac{\ln(p)}{p-1}$ diverges to $-\infty$. 
So for all sufficiently large $Y$ we have $\gamma(\mathbb{E}') < 1$.  We then find as before that
Cantor's work produces a non-zero polynomial $h(x) \in \mathbb{Q}[x]$ such that for all $v$
and all elements $z$ of the $v$-component of $\mathbb{E}'$ one has $|h(z)|_v \le 1$,
with $|h(z)| < 1$ if $v = \infty$.  In particular, $|h(z)|_p \le 1$ for all primes $p$ and all $z \in \mathbb{Z}_p \subset \overline{\mathbb{Z}}_p$.  So Corollary \ref{cor:Polyacor} shows $h(x)$ is an integral
combination of binomial polynomials such that $|h(z)| < 1$ if $z \in \mathbb{C}$ and $|z| \le N^\delta$.

\medbreak
\noindent {\bf Using the geometry of numbers to control the degree of auxiliary polynomials}
\medbreak
Minkowski's theorem says that if $L$ is a lattice in a Euclidean space $\mathbb{R}^n$ and $C$
is a convex symmetric subset of $\mathbb{R}^n$ of volume at least equal to $2^n$ times the
generalized index $[L:\mathbb{Z}^n]$, there must be a non-zero element of $L \cap C$.
To apply this to construct auxiliary polynomials, one takes $C$ to correspond to a suitably
bounded set of polynomials with real coefficients, and $L$ to correspond to those polynomials
with rational coefficients of the kind one is trying to construct. 

In the case at hand, suppose $1  \le r \in \mathbb{R}$.  Let $\mathbb{Z}[x]_{\le r}$ be the set of integral polynomials
		of degree $\le r$, and let $L_{\le r}$ be the $\Z$-span of $\{b_i(x): 0 \le i \le r, i \in \mathbb{Z}\}$.
		To show the first statement of Theorem \ref{thm:epsrsultpos}, it will suffice to show that if $c > 1$,
		then for sufficiently large $r  = N^\delta> 0$, there is a non-zero  $f(x) \in L_{\le cr}$
		such that $|f(z)| < 1$ for $z \in \mathbb{C}$ such that $|z| \le r$.  
		
		Let $m = \lfloor cr \rfloor$ be the largest integer less than or equal to $cr$.  By considering leading coefficients, we have
		$$\ln [L_{\le m} : \mathbb{Z}[x]_{\le m}]  = \ln \prod_{i = 0}^m i! = m^2 \ln(m)/2 \cdot (1 + o(1))$$
		where $o(1) \to 0$ as $m \to \infty$.
		Let $C$ be the set of polynomials with real coefficients of the form
		$$\sum_{i = 0}^m q_i (x/r)^i \quad \mathrm{with}\quad |q_i| \le 1/(m+2).$$
		We consider $C$ as a convex symmetric subset of $\mathbb{R}^{m+1}$
		by mapping a polynomial to its vector of coefficients.  Then
		$$\ln \mathrm{vol}(C) = (m+1)\cdot (\ln(2) - \ln(m+2)) - \sum_{i = 0}^m i \ln(r) = -\ln(r) m^2/2 \cdot (1 + o(1)).$$
		Since $\mathbb{Z}[x]_{\le m}$ maps to a lattice in $\mathbb{R}^{m+1}$ with covolume $1$,
		we find
		$$\ln \mathrm{vol}(C) - \ln \mathrm{vol}(\mathbb{R}^{m+1}/L_{\le m}) \ge (\ln(m) - \ln(r)) m^2/2 \cdot (1 + o(1)) = \ln(c) \cdot m^2/2 \cdot (1 + o(1)) .$$
		Since $\ln(c) > 0$, for sufficiently large $m$, the right hand side is greater than $2 \ln(m+1)$. Hence  Minkowski's
		Theorem produces a non-zero $f(x) \in L_{\le m}$ in $C$.  One has $$|f(z)| \le \sum_{i = 0}^m |z/r|^i/(m+2) < 1$$
		if  $z \in \mathbb{C}$ and $|z| < r$, so we have proved Theorem \ref{thm:epsrsultpos}.
\medbreak
\noindent {\bf An explicit construction}
\medbreak
	
		Theorem \ref{thm:epsrsultpos2} concerns the polynomials $b_{2t+1}(x+t)$ when $t > 0$ is an integer.
This polynomial takes integral values at integral $x$, so it is an integral combination of
the polynomials $b_i(x)$ with $0 \le i \le 2t+1$ by Polya's Theorem \ref{thm:Polya}.  To finish the proof of Theorem \ref{thm:epsrsultpos2},
it will suffice to show the following.  Let $q_0$ is the unique positive solution of
the equation (\ref{eq:qbound}),  and suppose $q > q_0$.  Let $D(r)$ be the closed disk
$D(r) = \{z \in \mathbb{C}:|z| \le r\}$.  We will show that if $r$ is sufficiently large, then
\begin{equation}
\label{eq:nicebound}
|b_{2t+1}(z+t)| < 1\quad \mathrm{if} \quad 2t \ge q r\quad \mathrm{and} \quad 
z \in D(r).
\end{equation}

        We have
$$b_{2t+1}(z + t) = \frac{\prod_{j = 0}^{2t} (z+t - j)}{(2t+1)!} = \frac{\prod_{j = -t}^{t} (z - j)}{(2t+1)!} = \pm \frac{z \cdot \prod_{j = 1}^t  (z^2 - j^2)}{(2t+1)!}$$
For $j \ge 0$ and $z \in D(r)$ we have 
$$|-r^2 - j^2| = r^2 + j^2 \ge |z^2 - j^2|.$$
So 
$$\mathrm{sup}(\{b_{2t+t}(z+t): z \in D(r)\}) = \frac{r \cdot \prod_{j = 1}^t  (r^2 + j^2)}{(2t+1)!}.$$
Taking logarithms gives
\begin{equation}
\label{eq:lnup}
\ln{\mathrm{sup} (\{b_{2t+t}(z+t): z \in D(r)\}) =\ln(r) + \sum_{j = 1}^t \ln(r^2 + j^2)  - \ln((2t+1)!)}.
\end{equation}
We now suppose $t \ge r$, so  $\xi = r/t \le 1$.  Then
 \begin{eqnarray}
 \label{eq:balance}
 \sum_{j = 1}^t \ln(r^2 + j^2) &=& t \ln(t^2) + t \cdot \frac{1}{t} \sum_{j = 1}^t \ln(\xi^2 + (j/t)^2)\nonumber\\
 &=& 2t \ln(t) + t \cdot \int_0^1 \ln(\xi^2 + s^2) ds + o(t)
 \end{eqnarray}
 as $t \to \infty$.  By integration by parts,
 \begin{equation}
 \label{eq:parts}
 \int \ln(\xi^2 + s^2) ds = s \ln(\xi^2 + s^2) - 2s + 2\xi \mathrm{arctan}(s/\xi).
 \end{equation}
 By Sterling's formula,
 \begin{equation}
 \label{eq:Sterling}
 \ln((2t+1)!) = (2t+1)\ln(2t+1) -( 2t + 1) + o(t) = 2t \ln(t) + 2t \ln(2) - 2t + o(t).
 \end{equation}
 Since $\ln(r) = o(t)$, we get from (\ref{eq:lnup}), (\ref{eq:parts}) and (\ref{eq:Sterling}) that
 \begin{equation}
\label{eq:lnup2}
\ln (\mathrm{sup} \{b_{2t+t}(z+t): z \in D(r)\}) =
t \cdot ( \ln(\xi^2 + 1) + 2\xi \mathrm{arctan}(\xi^{-1}) - 2\ln(2) ) + o(t).
\end{equation}
Writing $q = 2t/r = 2/ \xi \ge 2$ and multiplying both
 sides of (\ref{eq:lnup2})  by $q > 0$, we see that if 
 $$f(q) = q\ln \left(\frac{4}{q^2} + 1 \right) + 4 \mathrm{arctan}(q/2) - 2\ln(2) q < 0$$
 then for sufficiently large $t$ the supremum on the left in (\ref{eq:lnup}) is
 negative and we have the desired bound.  Here from $q \ge 2$ we have 
 $$f'(q) =\ln(1/q^2 + 1/4) \le \ln(1/2) < 0 < f(2) \quad \mathrm{and}\quad \lim_{q \to +\infty} f(q) = -\infty.$$
 So there is a unique positive real number $q_0$ with $f(q_0) = 0$, and $f(q) < 0$ for $q > q_0$.
 This establishes (\ref{eq:nicebound}) and 
  finishes the proof of part (ii) of Theorem \ref{thm:epsrsultpos}.

\subsection{Proof of Theorem~\ref{thm:epsresultneg}}

The proof of Theorem~\ref{thm:epsresultneg} uses a feedback procedure.
The feedback in this case is that if $N$ has no small prime
factor $p$, then for all small primes $p$ we can increase the set $E_p$.  This is described in more detail in \S \ref{s:showdontexist}.

Let $M$ be a positive integer and suppose $\epsilon > 0$. Suppose that there is a polynomial
		of the form 
		\begin{equation}
		\label{eq:hform}
		h(x) = \sum_{0 \le i,j \le M	} a_{i,j} b_i(x) b_j(f(x)/N)
		\end{equation}
		such that $a_{i,j} \in \mathbb{Z}$ and $|h(z)| < 1$ for all $z \in \mathbb{C}$
		such that $|z| \le N^{1/d + \epsilon}$.  We show that if $M$ satisfies one of the inequalities 
		involving $N$ in the statement of Theorem \ref{thm:epsresultneg}, then 
		$N$ must have a prime divisor bounded above by $M$. We will argue by contradiction.
		Thus we need to show that the following hypothesis cannot hold:
		\begin{hypothesis}
		\label{hyp:Contrary} No prime $p \le M$ divides $N$, and 
		 either (\ref{eq:Nprod}) holds or $1.48774 N^\epsilon \ge M \ge 319$.
		 \end{hypothesis}
		 
		 The point of the proof is to show that Hypothesis (\ref{hyp:Contrary}) leads to $h(x)$
		 having small sup norms on all components of an adelic set $\mathbb{E}$ which has capacity larger than $1$.
		 The reason that the hypothesis that no prime $p \le M$ divides $N$ enters into the
		 argument is that this guarantees that $f(z)/N$ will lie in the $p$-adic integers $\mathbb{Z}_p$
		 for all $z \in \mathbb{Z}_p$ when $p \le M$.  This will lead to being able to take the component
		 of $\mathbb{E}$ at such $p$ to be $\mathbb{Z}_p$.  
		 The $p$-adic capacity of $\mathbb{Z}_p$
		 is $p^{-1/(p-1)}$, as noted in Fact \ref{fact:capdisc}.  This turns out to be relatively large when one applies various results
		 from analytic number theory to get lower bounds on capacities.

		 To start a more detailed proof, let $p$ be a prime and suppose $0 \le i, j \le M$.  
		 
	\begin{lemma} 
	\label{lem:smallp} If $p \le M$ set $E_p = \mathbb{Z}_p$.  Then $|h(z)|_p \le 1$ if $z \in E_p$
	and the capacity $\gamma_p(E_p)$ equals $p^{-1/(p-1)} |N|_p$.
	\end{lemma}
	
	\begin{proof}	
		If $p  \le M$
		and $x \in \mathbb{Z}_p$, then $b_i(x) \in \mathbb{Z}_p$ since $\mathbb{Z}$
		is dense in $\mathbb{Z}_p$ and $b_i(x) \in \mathbb{Z}$ for all $x \in \mathbb{Z}$.
		Furthermore, $f(x)/N \in \mathbb{Z}_p$ for $x \in \mathbb{Z}_p$ since we
		have assumed $N$ is prime to $p$ and $f(x) \in \mathbb{Z}[x]$.  Therefore
		$b_j(f(x)/N) \in \mathbb{Z}_p$ for all $j$.  Since the 
		coefficients $a_{i,j}$ in (\ref{eq:hform}) are integers, we conclude $|h(z)|_p \le 1$.
		We remarked earlier in Fact \ref{fact:capdisc} that $\gamma_p(\mathbb{Z}_p) = p^{-1/(p-1)}$.  Since $p \le M$,
		we have supposed that $p$ does not divide $N$.  So $|N|_p = 1$, and we get
		$\gamma_p(E_p) = \gamma(\mathbb{Z}_p) = p^{-1/(p-1)} |N|_p$.
\end{proof}
	
	\begin{lemma} 
	\label{lem:bigp} If $p > M$ set $E_p =f^{-1}(N\overline{\mathbb{Z}}_p)$.  Then $|h(z)|_p \le 1$ if $z \in E_p$
	and $\gamma_p(E_p) = |N|_p^{-1/p}$.
	\end{lemma}
     
\begin{proof}
	      We first note that $0 \le i, j  \le M < p $ implies that $|i!|_p  = |j!|_p = 1$.
	     Recall that $\overline{\mathbb{Z}}_p = \{x \in \Qpbar: |x|_p \le 1\}$.  If $x \in f^{-1}(N
	     \overline{\mathbb{Z}}_p)$ then $x \in \overline{\mathbb{Z}}_p$ since $f(x)$
	     is monic with integral coefficients.  So  
	     $$|b_i(x)|_p = \frac{|x\cdot (x-1) \cdots (x - i +1)|_p}{|i!|_p} \le 1$$
	     and 
	     $$ |b_j(f(x)/N)|_p  = \frac{|f(x)/N\cdot (f(x)/N-1) \cdots (f(x)/N - j +1)|_p}{|j!|_p}  \le 1 $$
	     since $x - k$ and $ f(x)/N - k$ lie in $ \overline{\mathbb{Z}}_p$ for all integers $k$ and $|i!|_p = |j!|_p = 1$.
	      Because
	     the $a_{i,j}$ in (\ref{eq:form}) are integral, we conclude $|h(z)|_p \le 1$ if $z \in E_p =f^{-1}(N\overline{\mathbb{Z}}_p)$. The capacity $\gamma_p(E_p)$ is $|N|_p^{-1/p}$ by 
	     Fact \ref{fact:cappreimage}.	     
	     \end{proof}
	     
\begin{lemma}
\label{lem:pinf} Set $E_\infty = \{z \in \mathbb{C}: |z| \le N^{1/d + \epsilon} \}$. Then $|h(z)|_\infty < 1$
if $z \in E_\infty$ and $\gamma_\infty(E_\infty) = N^{1/d + \epsilon}$.  
\end{lemma}
\begin{proof}  This first statement was one of our hypotheses on $h(x)$, while 
 $\gamma_\infty(E_\infty) = N^{1/d + \epsilon}$ by Fact \ref{fact:capdisc}. 
\end{proof}
     
	   We conclude from these Lemmas and Fact \ref{fact:capprod} that 	when
	   \[\mathbb{E} = \prod_p E_p \times E_\infty\]
	   we have    
	   	     \begin{equation}
	     \label{eq:capagain}
	     \gamma(\mathbb{E}) = \left(\prod_{p \le M} p^{-1/(p-1)}\right) \times \left(\prod_{all \ p} |N|_p^{1/d}\right) \times N^{1/d + \epsilon} = \left(\prod_{p \le M} p^{-1/(p-1)}\right) N^\epsilon.
	     \end{equation}
	      Here
	      $$\ln \left(\prod_{p \le M} p^{-1/(p-1)}\right) = - \sum_{p \le M} \frac{ \ln(p)}{p-1}$$
	      and it follows from \cite[Theorem 6, p. 70]{Rosser} that if $M \ge 319$ then
	      \begin{eqnarray}
	      \label{eq:amazing}
	      -\sum_{p \le M}  \frac{ \ln(p)}{p-1} &=& -\sum_{p \le M} \frac{\ln(p)}{p} - \sum_{p \le M}  \frac{\ln(p)}{p(p-1)} \nonumber \\
	      &\ge& -\sum_{p \le M} \frac{\ln(p)}{p} - \sum_p \sum_{n=2}^\infty \frac{\ln(p)}{p^n}\nonumber \\
	      &\ge&- \ln(M) + \gamma - \frac{1}{\ln (M)}
	      \end{eqnarray} 
	      where $\gamma = 0.57721...$ is Euler's constant.
	      
Hence (\ref{eq:capagain}) gives
\begin{equation}
\label{eq:conditions}
\ln(\gamma(\mathbb{E})) = -\sum_{p \le M} \frac{ \ln(p)}{p-1} + \epsilon \ln(N) \ge  
- \ln(M) + \gamma - \frac{1}{\ln (M)} + \epsilon \ln(N).
\end{equation}
The right hand side is positive if
\begin{equation}
\label{eq:alright}
 N^\epsilon \cdot e^{\gamma - 1/\ln(M)} > M.
 \end{equation}
Since we assumed $M \ge 319$, we have $e^{\gamma - 1/\ln(M)} \ge 1.497445...$ and 
so (\ref{eq:alright}) will hold if
\begin{equation}
\label{eq:alright2}
1.48744 \cdot N^\epsilon > M
\end{equation}
In any case, if the left hand side of (\ref{eq:conditions}) is positive then 
 $\gamma(\mathbb{E}) > 1$.  However, we have shown that $h(x)$ is a non-zero polynomial
	      in $\mathbb{Q}[x]$ such that $|h(x)|_v \le 1$ for all $v$ when $x \in E_v$
	      with strict inequality when $v = \infty$.   By Cantor's Theorem~\ref{fact:aux}, such an $h(x)$ cannot exist
	      because $\gamma(\mathbb{E}) > 1$.  The contradiction shows that 
	      Hypothesis \ref{hyp:Contrary} cannot hold, and this  completes the proof of
	      Theorem \ref{thm:epsresultneg}.

\section{A field guide for capacity-theoretic arguments}
\label{s:fieldguide}

The proofs in \S\ref{sec:capacity} and \S\ref{s:improve} illustrate how capacity theory can be used to show the nonexistence and existence of polynomials with certain properties.  This paper is a first step toward building a more general framework to apply capacity theory to cryptographic applications.
In this section, we step back and 
summarize how capacity theory can be used in general to show either that auxiliary polynomials
with various desirable properties do or do not exist.  

The procedure for
applying capacity theory to such problems allows for feedback between the type of polynomials one
seeks and the computation of the relevant associated capacities.  If it turns out that the capacity
theoretic computations are not sufficient for a definite conclusion, they may suggest additional
hypotheses either on the polynomials or on auxiliary parameters which would be useful
to add in order to arrive at a definitive answer.  They may also suggest some alternative proof
methods which will succeed even when capacity theory used as a black box does not.  

\subsection{Showing auxiliary polynomials exist}
\label{s:showexist}

To use capacity theory to show that polynomials $h(x) \in \mathbb{Q}[x]$ with certain properties exist, 
one can follow these steps:

\begin{enumerate}[align=left]
\item[{\bf Step 1.}] State the conditions on $h(x)$ which one would like to achieve.  These can be of an intrinsic
or an extrinsic nature.  

\begin{enumerate}
\item{} Intrinsic conditions have the following form:
\begin{enumerate}
\item[(i)] For each prime $p$, one should give a subset $E_p$ of $\overline{\mathbb{Q}}_p$.
For all but
finitely many $p$, $E_p$ must be the set $\overline{\mathbb{Z}}_p$.   
\item[(ii)] One should give a subset $E_\infty$ of $\mathbb{C}$.
\item[(iii)]  The set of polynomials $h(x) \in \mathbb{Q}[x]$ one seeks are all polynomials such that
$|h(z)|_p \le 1$ for all primes $p$ and all $z \in E_p$ and $|h(w)| < 1$ if $w \in E_\infty$.
\end{enumerate}
  \medbreak
\item{} To state 
conditions on $h(x)$ extrinsically, one writes down the type of polynomial expressions one allows.
For example, one might require $h(x)$ to be an integral combination of integer multiples of 
specified polynomials, e.g. monomials in $x$ as in Theorem  \ref{thm:exist}.  Suppose one uses such
an extrinsic description, and one is trying to show the existence of $h(x)$ of this form using capacity
theory.  It is then necessary to come up with an intrinsic description of the above kind with the property
that any $h(x)$ meeting the intrinsic conditions must have the required extrinsic description.
We saw another example of this in \S\ref{s:improve} on binomial polynomials;  see also Step 5 below.
\end{enumerate}
\item[{\bf Step 2.}] Suppose we have stated an intrinsic condition on $h(x)$ as in parts (i), (ii) and (iii)
of Step 1(a).  One then needs to check that the adelic set $\mathbb{E} = \prod_p E_p \times E_\infty$
satisfies certain standard hypotheses specified in \cite{Cantor} and \cite{Rumely1,Rumely2}.  These ensure
that the  capacity 
\begin{equation}
\label{eq:maybe}
\gamma(\mathbb{E}) = \prod_p \gamma_p(E_p) \cdot \gamma_\infty(E_\infty)
\end{equation}
is well defined.  One then needs to employ  \cite{Cantor} and \cite{Rumely1,Rumely2} to find an upper bounds
the $\gamma_p(E_p)$, on $\gamma_\infty(E_\infty)$ and then on $\gamma(\mathbb{E})$.  This may also require
results from analytic number theory concerning the distribution of primes.  When using this method 
theoretically, there may be an issue concerning the computational complexity of finding such
upper bounds.  However, if $E_p$ and $E_\infty$ have a simple form (e.g. if they are disks),
explicit formulas are available.  Notice that the requirement in part (i) of Step 1 that $E_p = \overline{\mathbb{Z}}_p$ for all but finitely many $p$ forces $\gamma_p(E_p) = 1$
for all but finitely many $p$. So the product on the right side of  (\ref{eq:maybe}) is well defined as long
as $\gamma_\infty(E_\infty)$ and $\gamma_p(E_p)$ are for all $p$.  
\item[{\bf Step 3.}] If the computation in Step 2 shows $\gamma(\mathbb{E}) < 1$, capacity
theory guarantees that there is some non-zero polynomial $h(x) \in \mathbb{Q}[x]$ which satisfies
the bounds in part (iii) of Step 1. However, one has no information at this point
about the degree of $h(x)$.   
\item[{\bf Step 4.}]  Suppose that Step 2 shows $\gamma(\mathbb{E}) < 1$ and that we want to
show there is an $h(x)$ as in Step 3 satisfying a certain bound on its degree.  There are three levels
of looking for such degree bounds. 
\begin{enumerate}
\item[a.] The most constructive method is to present an explicit construction of
an $h(x)$ which one can show works.  We did this in the previous section in the case of integral
combinations of binomial polynomials.  
\item[b.] The second most constructive method is to convert the existence
of $h(x)$ into the problem of finding a short vector in a  suitable lattice of polynomials and to
apply the LLL algorithm.  One needs to show that the LLL criteria are met once one considers
polynomials of a sufficiently large degree, and that a short vector will meet the intrinsic criteria
on $h(x)$.  We will return in later papers to the general question of when $\gamma(\mathbb{E}) < 1$
implies that there is a short vector problem whose solution via LLL will meet the intrinsic
criteria.  This need not always be the case.  The reason is that in the geometry of numbers,
one can find large complicated convex symmetric sets which are very  far from being generalized
ellipsoids.  However, in practice, the statement that  $\gamma(\mathbb{E}) < 1$ makes it
highly likely that the above LLL approach will succeed.   
\item[c.] Because of the definition of sectional capacity in \cite{ChinburgCap} and \cite{RumelyLauVarley},
the following approach is guaranteed
to succeed by $\gamma(\mathbb{E}) < 1$.  Minkowski's Theorem in the geometry
of numbers will produce (in a non-explicit manner) a polynomial $h(x)$ of large degree $m$
which meets the intrinsic criteria.  One can estimate how large $m$ must be by computing
certain volumes and generalized indices.  We illustrate such computations in \S\ref{s:improve}
in the case of intrinsic conditions satisfied by integral combination of binomial polynomials.
\end{enumerate} 
\item[{\bf Step 5.}]  It can happen that the most natural choices for $E_p$ and $E_\infty$ in step 1
above do not satisfy all the criteria for the capacity of $\mathbb{E} = \prod_p E_p \times E_\infty$
to be well defined.  One can then adjust these choices slightly.  To obtain more control on
the degrees of auxiliary functions, one can try an explicit Minkowski argument of the kind 
use in the proof of  the positive result concerning
integral combinations of binomial polynomials in Theorem \ref{thm:epsrsultpos} above.
\end{enumerate}

 \subsection{Showing auxiliary polynomials do not exist}
\label{s:showdontexist}

To use capacity theory to show that polynomials $h(x) \in \mathbb{Q}[x]$ with certain properties do not exist, 
one can follow these steps:

\begin{enumerate}[align=left]
\item[{\bf Step 1.}] Specify the set of properties you want $h(x)$ to have.  Then show that the following is true for every $h(x)$ with these properties:
\begin{enumerate}
\item[(i)] For each prime $p$, exhibit a set $E_p$ of $\overline{\mathbb{Q}}_p$
such that $|h(z)|_p \le 1 $ if $z \in E_p$.  For all but
finitely many $p$, $E_p$ must be the set $\overline{\mathbb{Z}}_p$.   
\item[(ii)] Exhibit a closed subset $E_\infty$ of $\mathbb{C}$ such that  $|h(z)| < 1$ if $z \in E_\infty$.
\end{enumerate}
It is important that $h(x)\in \mathbb{Q}[x]$ with the desired properties meet the criteria in (i) and (ii).
\item[{\bf Step 2.}]  As before, one needs to check that the adelic set $\mathbb{E} = \prod_p E_p \times E_\infty$
satisfies certain standard hypotheses specified in \cite{Cantor} and \cite{Rumely1,Rumely2}.  These ensure
that the  capacity 
\begin{equation}
\label{eq:maybe2}
\gamma(\mathbb{E}) = \prod_p \gamma_p(E_p) \cdot \gamma_\infty(E_\infty)
\end{equation}
is well defined.  One then needs to find a lower bound on $\gamma(\mathbb{E})$
using lower bounds on the $\gamma_p(E_p)$ and on $\gamma_\infty(E_\infty)$.
 One may also require
information from analytic number theory, e.g. on the distributions of prime numbers less than
a given bound.  
\item[{\bf Step 3.}] If the computation in Step 2 shows $\gamma(\mathbb{E}) > 1$, capacity
theory guarantees that there is no non-zero polynomial $h(x) \in \mathbb{Q}[x]$ which satisfies
the intrinsic conditions (i) and (ii) of Step 1. This means there do not exist of polynomials 
$h(x)$ having the original list of properties.
\item[{\bf Step 4.}]  Suppose that in Step 3, we cannot show $\gamma(\mathbb{E}) > 1$ due to the
fact that the sets $E_p$ and $E_\infty$ in Step 1 are not sufficient large.  One can now change
the original criteria on $h(x)$, or take into account some additional information, to try to enlarge
the sets $E_p$ and $E_\infty$ for which Step 1 applies.  We saw in the previous section
how this procedure works in the case of integral combinations of certain products of
binomial polynomials.  For example, if one assumes that certain other parameters (e.g. the modulus of a congruence)
have no small prime factors, one can enlarge the sets $E_p$ in Step 1 which are associated
to small primes. 
\end{enumerate}

\section{Conclusion}

In this work, we drew a new  connection between two disparate research areas: lattice-based techniques for cryptanalysis and capacity theory.
This connection has benefits for researchers in both areas.

\begin{itemize}
\item
\textbf{Capacity Theory for cryptographers:} We have shown that techniques from capacity theory 
can be used to show that the bound obtained by Coppersmith's method in the case of univariate polynomials is 
optimal for a broad class of techniques.  This implies that the best available class of techniques for solving these types of problems cannot be extended.  This has implications for cryptanalysis, and the tightness of cryptographic security reductions.

\item
\textbf{Cryptography for capacity theorists:} Capacity theory provides a method for calculating the conditions under 
which certain auxiliary polynomials exist.  Coppersmith's method provides an efficient algorithm for \emph{finding} 
these auxiliary polynomials.  Until this time, capacity theory has not addressed the computational complexity
of producing auxilary functions when they do exist.
\end{itemize}

Concretely, we used capacity theory to answer three questions of Coppersmith in \cite{CoppersmithFinding}

\begin{enumerate}
	\item
		Can the exponent $1/d$ be improved (possibly through improved lattice reduction techniques)? No, the desired auxiliary polynomial simply does not exist.
	\item
		Does restricting attention to the real line $[-N^{-1/d},N^{1/d}]$ instead of the complex disk $|z| \le N^{1/d}$ 
		improve the situation?  No.
	\item
		Does considering lattices based on binomial polynomials improve the situation?  No, these lattices have the desired auxiliary polynomials, but for RSA moduli, their degree is too large to be useful.
\end{enumerate}

Since Coppersmith's method is one of the primary tools in asymmetric cryptanalysis, these results give an indication of the security of many factoring-based cryptosystems.

This paper lays a foundation for several directions of future work.  
Coppersmith's study of small integral solutions of equations in two variables and bivariate
equations modulo $N$  \cite{C97} is related to capacity theory on curves, as developed by Rumely in \cite{Rumely1,Rumely2}.  The extension of
Coppersmith's method to multivariate equations \cite{jochemsz,jutla} is connected to capacity theory on higher dimensional
varieties, as developed in \cite{ChinburgCap}, \cite{RumelyLauVarley} and \cite{CrelleCap}.   Multivariate problems
raise deep problems in arithmetic geometry about the existence of finite morphisms to 
projective spaces which are bounded on specified archimedean and non-archimedean sets.
Interestingly, Howgrave-Graham's extension of
Coppersmith's method to find small roots of modular equations modulo \emph{unknown} moduli \cite{HG01,M10} appears to pertain to joint capacities of many adelic sets, a topic which has not been developed to our knowledge   in the capacity theory literature.
It is an intriguing question whether capacity theory can be extended to help us understand the limitations of these more general variants of Coppersmith's method.


\section*{Acknowledgements}

This material is based upon work supported by the National Science Foundation under grants CNS-1513671, DMS-1265290, DMS-1360767, CNS-1408734, CNS-1505799, by the Simons Foundation under fellowship 338379, and a gift from Cisco.

\bibliographystyle{alpha}
\bibliography{crypto,capacity}

\end{document}